\documentclass[11pt]{article}
\usepackage{amsmath,amsfonts,amsthm,amssymb,algorithm,algorithmic}
\usepackage[numbers]{natbib}
\usepackage{setspace}
\usepackage{Tabbing}
\usepackage{fancyhdr}
\usepackage{lastpage}
\usepackage{extramarks}
\usepackage{chngpage}
\usepackage{soul,color}
\usepackage{graphicx,float,wrapfig}
\usepackage{rotating}
\usepackage{cancel}
\usepackage{pifont}
\usepackage{MnSymbol}
\usepackage{tikz}
\usepackage{subcaption}

\newtheorem{theorem}{Theorem}

\newcommand{\bc}{\begin{center}}
\newcommand{\ec}{\end{center}}
\newcommand{\beq}{\begin{equation}}
\newcommand{\eeq}{\end{equation}}
\newcommand{\benum}{\begin{enumerate}}
\newcommand{\eenum}{\end{enumerate}}
\newcommand{\bea}{\begin{eqnarray*}}
\newcommand{\eea}{\end{eqnarray*}}
\newcommand{\ba}{\begin{array}}
\newcommand{\ea}{\end{array}}

\newcommand{\red}{\textcolor{red}}
\newcommand{\blue}{\textcolor{blue}}

\begin{document}
\title{Local rewiring algorithms to increase clustering and grow a small world}
\author{Jeff Alstott,\thanks{Massachusetts Institute of Technology/Singapore University of Technology and Design, alstott@mit.edu}
\and Christine Klymko, \thanks{Lawrence Livermore National Laboratory, klymko1@llnl.gov}
\and Pamela B. Pyzza, \thanks{Ohio Wesleyan University, pbpyzza@owu.edu}
\and Mary Radcliffe \thanks{Carnegie Mellon University, mradclif@math.cmu.edu}
}
\date{}
\maketitle
Many real-world networks have high clustering among vertices: vertices that share neighbors are often also directly connected to each other. A network's clustering can be a useful indicator of its connectedness and community structure. Algorithms for generating networks with high clustering have been developed, but typically rely on adding or removing edges and nodes, sometimes from a completely empty network. Here, we introduce algorithms that create a highly clustered network by starting with an existing network and rearranging edges, without adding or removing them; these algorithms can preserve other network properties even as the clustering increases. They rely on local rewiring rules, in which a single edge changes one of its vertices in a way that is guaranteed to increase clustering. This greedy step can be applied iteratively to transform a random network into a form with much higher clustering. Additionally, the algorithms presented grow a network's clustering faster than they increase its path length, meaning that network enters a regime of comparatively high clustering and low path length: a small world. These algorithms may be a basis for how real-world networks rearrange themselves organically to achieve or maintain high clustering and small-world structure. 

\section{Introduction}
\label{sec:intro}

Many complex systems are organized as networks, including social networks, gene interaction networks, the world wide web, and beyond \cite{BoLaMoCh06,BrEr05}.  The majority of complex networks from real-world systems have statistical properties that separate them from the random graphs studied more classically. These properties include a skewed or power-law degree distribution \cite{BaAl99,CaRoNe09}, a high clustering coefficient \cite{BaWe99,LuPe49}, and a low path length, creating a small world of connectivity \cite{WattsStrogatz1998}.

A network's clustering coefficient is a measure of how often two connected nodes will have neighbors in common. More formally, the clustering coefficient relates a network's number of triangles (3-cycles) to the possible number of triangles (given by the total number of wedges in the network). There are two common ways to define a clustering coefficient, which are related but different. The first, often called the global clustering coefficient, involves a global count of the number of triangles and possible triangles in a network.  The second involves an average of local counts. Formal definitions can be found in Section~\ref{sec:def}.  Real-world networks tend to have (relatively) high clustering coefficients under both definitions due to the presence of communities: if two nodes are in the same community, it not only increases the likelihood that there is an edge between them, but also that they have a common neighbor.  Due to this, the clustering coefficient is a good measure of the well-connectedness of a network and can indicate the presence or lack of strong community structure \cite{LaFo09,RaCaCeLoPa04}.

Closely related to the concept of the clustering coefficient of a network is the idea of the {\em small world} property found in many real world networks.  Informally, the small world property states that the average distance between any two nodes in a given network is relatively small and, in the case of an evolving network, the distance grows more slowly than the network grows.  In networks that are not fully connected, having a short average distance is opposed to having a high clustering coefficient. This is due to the fact that, in order for nodes in distant parts of the network to have short paths between each other, there must exist some long range edges between communities.  However, adding long range edges increases the number of possible triangles in a network without creating any new triangles, since nodes in different communities are unlikely to have neighbors in common, thus lowering the clustering coefficient.  There has been research focusing on how to jointly maximize these two properties under various constraints \cite{barmpoutis2010networks,ZhYaWa05} as well as on understanding processes which may lead to the dual development of both properties in both real world and generated networks \cite{ClMo03,HoKi02,LeGoKaKi04,newman2009random,PeBiTaCh09}.

Most research on generating networks with high clustering or low path length has focused on building these networks from scratch. However, in many real systems a network is already in existence. One may want to minimally rearrange existing edges in order to enhance or reduce certain network properties. The goal of these rewirings include increasing network robustness \cite{begrliri05,JiLiAn14,JiLiGu13,KoPuGa13,LoDaHeTo13,wuellner2010resilience,zhou2014memetic}, communicability \cite{ArBE15}, synchronizablity \cite{LiQiZhYu10} or algebraic connectivity \cite{SyScGr13}. Similar work has been done on the careful addition of edges to improve network features, including conductance \cite{zhou2016faster, papagelis2015refining}, closeness centrality \cite{parotsidis2016centrality, crescenzi2016greedily}, and path lengths \cite{parotsidis2016centrality, papagelis2015refining}. A common feature of many of these rewirings or edge additions is the focus on local properties of the network, such as an individual node and its neighborhood, in order to ultimately impact global properties of a network, such as conductance, robustness, and connectivity. Various algorithms have been proposed in the aformentioned papers to make the best use of such rewirings or edge additions.

Here we present edge rewiring algorithms which increases the clustering coefficient of a given network while minimally impacting other network properties, specifically degree distribution and average path length. We provide proofs that these algorithms monotonically increase the global clustering coefficient. We present basic notation and definitions in Section~\ref{sec:def}. We then describe the algorithms in Section~\ref{sec:algorithms}, along with  theoretical results about the effect of these algorithms on clustering. We numerically simulate the algorithms on model networks in Section~\ref{sec:experiments} and present results on real-world networks in Section~\ref{sec:real_experiments}. We discuss the implications of these algorithms in Section~\ref{sec:conclusions}.

\section{Notation and definitions}
\label{sec:def}
Let $G=(V, E)$ be a graph with a set of vertices (also referred to as nodes) $V$, $|V| = n$, and edge set $E = \{u\sim v\ | \ u,v \in V $ are connected by an edge$\}$, and write $\bar{E}$ for the complement of $E$ in ${V\choose 2}$. For a vertex $v\in V(G)$, $d_v$ denotes the {\em degree} of vertex $v$. We define the {\it neighborhood} of $v$ in $G$ to be $N(v) = \{u\in V(G)\ | \ u\sim v\}$. To simplify notation, we shall denote by $\overline N(v)$ the complement of $N(v)$ in $V(G)\backslash\{v\}$; that is, $\overline N(v)$ is the set of vertices in $G$ other than $v$ itself to which $v$ is not adjacent. Given two vertices $x, y$, define their {\it common neighborhood} to be $N(x, y)=\{u\in V(G)\ | \ u\sim x\hbox{ and }u\sim y\}$. For a set $S\subset V$, define $e(S)$ to be the number of edges in $G$ having both endpoints in $S$.

Define $N_p(G)$ and $N_t(G)$ to be the number of length-two paths and number of triangles in the graph, respectively. Note that \[N_p(G)=\sum_{x, y\in V(G)} |N(x, y)|,\] and \[3N_t(G)=\sum_{\{x, y\}\in E(G)} |N(x, y)|.\] We define the clustering coefficient of $G$ to be $$C(G)=3N_t(G)/N_p(G).$$

The above version of the clustering coefficient is sometimes referred to as the {\em global clustering coefficient} and we shall use that language at times in this work. Another measure of clustering, known as the {\em local average clustering coefficient}, is also in common usage. For a vertex $v\in V(G)$, let 
\[c_v = \frac{2e(N(v))}{d_v(d_v-1)}.\]
Here, we have that ${d_v\choose 2}$ is the number of length-two paths having $v$ at the center; that is, it is a measure of the number of triangles in which $v$ could be involved. Thus, $c_v$ can be seen as the proportion of these triangles that actually exist in the graph, compared to how many triangles are possible involving $v$. We then define the local average clustering coefficient to be $$c(G) = \frac{1}{n}\sum_{v\in V}c_v.$$ 

Although the global clustering coefficient and the local average clustering coefficient often produce similar results regarding clustering in the graph, they are not equivalent measures \cite[p.~83]{Es12}. The primary difference between the two is the amount of emphasis placed on lower degree vertices. In the local average clustering coefficient, each vertex is treated with equal weight and, hence, vertices of very low degree, for which $c_v$ might be unusually high or low, have a much more substantial impact on the constant than those of high degree. In the global version, this impact is avoided by considering the graph as a whole, rather than any specific local structure.

The {\it adjacency matrix} associated with a graph $G$ is given by $A = a(u,v)$ with 
$$a(u,v) = \left\{\begin{array}{ll}
1,& \textnormal{ if } u \sim v,\\
0, & \textnormal{ else. }
\end{array}\right .
$$

Given two sets $A, B\subset X$, we use the notation $A\triangle B$ to denote the symmetric difference between $A$ and $B$; that is, 
\[A\triangle B = (A\backslash B)\cup (B\backslash A).\]

\section{Algorithms}
\label{sec:algorithms}

In this section, we fully describe the algorithms used to rewire a graph to increase its clustering coefficient and prove that these algorithms are monotonic with respect to the global clustering coefficient. Further, we examine some of the local optima of these algorithms and consider the impact of the rewiring procedure on the local average clustering coefficient.  Finally, we briefly discuss their computational complexity.

\subsection{Algorithm Description}

To begin, let us examine the algorithms in question. The fundamental idea of the algorithms we employ here is as follows. We wish to identify a set of vertices $\{x, y, v\}$ such that $\{x, y\}$ are not adjacent, $\{x, v\}$ are adjacent, and further, that if we rewired the edge $\{x, v\}$ to $\{x, y\}$, the clustering coefficient of the graph will improve. Fundamentally, if this move is considered from the perspective of the vertex $x$, the goal is for $x$ to replace less valuable edges (with respect to clustering) with more valuable edges. This process is then iterated so that the clustering coefficient increases upon each edge movement. 

We shall consider two different, but similar algorithms. Both of these are based on the iterated protocol described above. The first version of the protocol is detailed in Rewiring Algorithm \ref{Algone} below. Essentially, the protocol is as follows. We first choose a nonedge $\{x, y\}$ whose addition to the graph $G$ would increase the number of triangles in $G$ as much as possible. We then look among the existing incident edges for an edge whose removal would decrease the number of triangles in $G$ as little as possible. If certain degree considerations are met, we then rewire this edge to $\{x, y\}$ to form a new graph $G'$.  More details about the degree requirements can be found in Thm.~\ref{thm:rewiring} and its proof. The edge movement is akin to swinging a door from one doorway to another. The algorithm's process can be described as finding the best open doorway and swinging a door toward it. This is illustrated in Figure~\ref{Fig:Firstalg}. 

\begin{algorithm}[H]
\floatname{algorithm}{``Swing Toward Best" Rewiring Algorithm}
\caption{}
\label{Algone}
\begin{algorithmic}[1]
\STATE Find a pair $\{x, y\}\in \bar{E}$ with the maximum value of $|N(x, y)|$
\STATE For $v\in N(x)$, define $f_v= |N(x, v)|$, and for $v\in N(y)$, define $f_v=|N(y, v)|$. Choose the vertex $v\in N(x)\triangle N(y)$ with minimum $f_v$; given a tie, choose $v$ to have the highest possible degree. WLOG, suppose $v\sim x$.
\STATE If $|N(x, v)|\geq |N(x, y)|$, return to step 1, and eliminate the edge $\{x, y\}$ from consideration.
\STATE If $d_v>d_y$, rewire the edge $vx$ to the edge $yx$ to form $G'$. 
\STATE If $d_v\leq d_y$, return to step 2, and eliminate the vertex $v$ from consideration.
\STATE If no neighbor in $N(x)\triangle N(y)$ satisfies the requirements, return to step 1 and choose a different pair of nonadjacent vertices.
\end{algorithmic}
\end{algorithm}

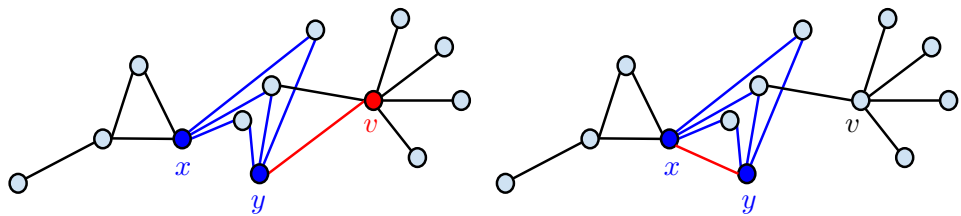
\begin{figure}[H]
\centering
\begin{tikzpicture}
\node[inner sep=0pt] (pic) at (0,0) {\includegraphics[width = 2.5 in]{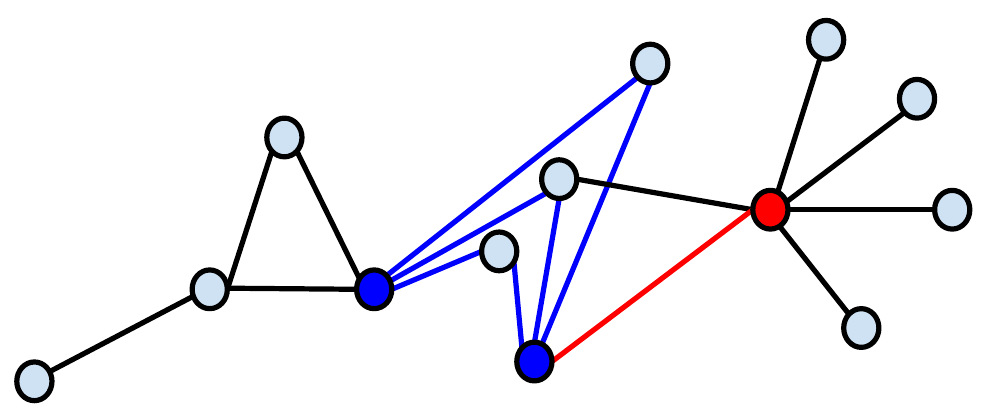} \includegraphics[width = 2.5 in]{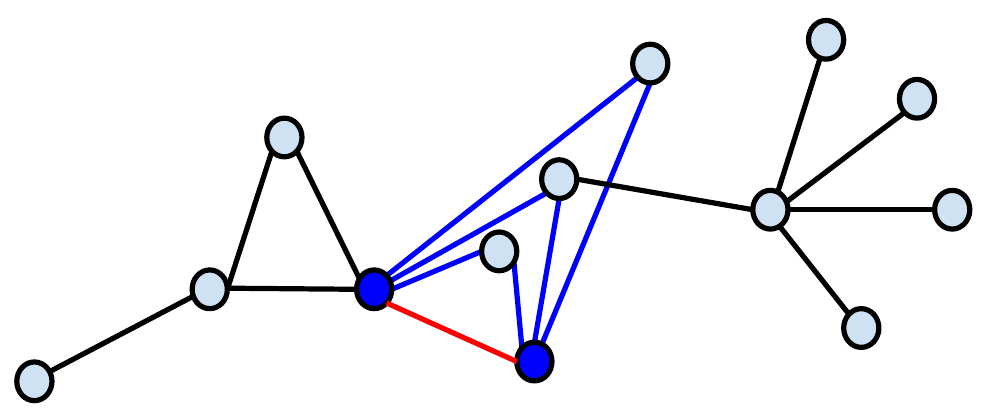}};
\node at (-4,-0.925) {\blue{$x$}};
\node at (2.5,-0.925) {\blue{$x$}};
\node at (-3,-1.4) {\blue{$y$}};
\node at (3.525,-1.4) {\blue{$y$}};
\node at (-1.5,-0.35) {\red{$v$}};
\node at (4.9,-0.35) {$v$};
\end{tikzpicture}
\caption{A graph $G$ being rewired using Rewiring Algorithm \ref{Algone}. Here, the two blue vertices are $x$ and $y$. Note that in the original graph (left), $x$ and $y$ have the largest number of common neighbors among nonedges of the graph. We then choose $v$ as the red vertex, as its rewiring would destroy the fewest number of triangles (open the fewest closed doors), and it has maximum degree among such vertices. After verifying that all the degree considerations are met, we rewire the edge $xv$ to form the new graph (right).\\}
\label{Fig:Firstalg}
\end{figure}

For the second version of the protocol, described below in Rewiring Algorithm \ref{Algtwo}, we take a slightly different approach to choosing an edge to rewire. In the first version, we first seek an open doorway, having as many incomplete triangles as possible. For the second version, we take the opposite approach, and choose an edge whose removal would destroy the fewest number of triangles in $G$ as possible; that is to say, the algorithm finds the closed door that is least beneficial to clustering and swings the door away to a more useful doorway. This approach can be much faster when the graphs in question are sparse; the remainder of the procedure is essentially the same as the first version of the algorithm. 

\begin{algorithm}[H]
\floatname{algorithm}{``Swing Away from Worst" Rewiring Algorithm}
\caption{}
\label{Algtwo}
\begin{algorithmic}[1]
\STATE Find a pair $\{x, v\}\in E$ with the minimum value of $|N(x, v)|$
\STATE For $y\in \overline{N}(x)$, define $f_y= |N(x, y)|$, and for $y\in \overline N(v)$, define $f_y=|N(y, v)|$. Choose a vertex $y\in \overline N(x)\triangle \overline N(v)$ with minimum $f_y$; given a tie, choose $y$ to have the lowest possible degree. WLOG, suppose $y\sim x$.
\STATE If $|N(x, y)|< |N(x, v)|$, return to step 1,  and eliminate the edge $\{x, v\}$ from consideration.
\STATE \label{Deg1}If $d_y<d_v$, rewire the edge $vx$ to the edge $yx$ to form $G'$. 
\STATE \label{Deg2} If $d_y\geq d_v$, return to step 2, and eliminate the vertex $y$ from consideration.
\STATE If no neighbor in $\overline N(x)\triangle \overline N(v)$ satisfies the requirements, return to step 1 and choose a different pair of nonadjacent vertices.
\end{algorithmic}
\end{algorithm}

\begin{figure}[H]
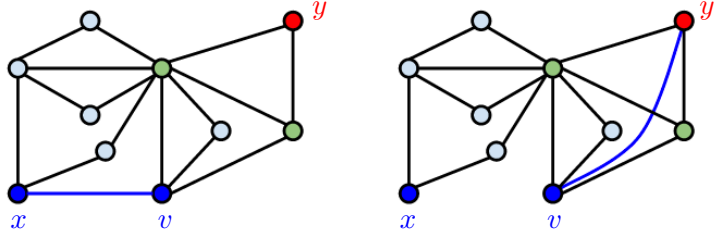

\centering
\begin{tikzpicture}
\node[inner sep=0pt] (pic) at (0,0) {\includegraphics[width = 2 in]{p21.pdf} \includegraphics[width = 2 in]{p22.pdf}};
\node at (-4.55,-1.6) {\blue{$x$}};
\node at (0.625,-1.6) {\blue{$x$}};
\node at (-2.6,-1.6) {\blue{$v$}};
\node at (2.575,-1.6) {\blue{$v$}};
\node at (-0.55,1.2) {\red{$y$}};
\node at (4.6,1.2) {\red{$y$}};
\end{tikzpicture}
\caption{A graph $G$ being rewired using Rewiring Algorithm \ref{Algtwo}. Here, the two blue vertices are $x$ and $v$. Note that in the original graph (left), $x$ and $v$ have the fewest number of common neighbors among edges of the graph; namely, $x$ and $v$ are not involved in any common triangles. We then choose $y$ as the red vertex, noting that rewiring $xv$ to $vy$ will add two triangles at the green vertices. After verifying that all the degree considerations are met, we rewire the edge to $vy$ to form the new graph (right).\\}\label{Fig:Secondalg}
\end{figure}

First, let us verify that the action of ``swinging a door,'' regardless of the algorithm itself, is monotonic with respect to clustering coefficient.

\begin{theorem}
\label{thm:rewiring}
Any individual rewiring performed satisfying the conditions of the algorithms above will strictly increase the clustering coefficient.
\end{theorem}

\begin{proof}

Suppose we have three vertices $x, y, v$ such that
\begin{itemize}
\item $\{x, y\}\in \bar{E}$, $\{x, v\}\in E$
\item $|N(x, y)|> |N(x, v)|$
\item $d_v>d_y$.
\end{itemize}
Then the rewiring of edge $xv$ to edge $xy$ is permitted according to either of the algorithms described above. We note that this situation accommodates both versions of the protocol, although the vertex labels are changed for the second version. Let $G'=(G\backslash\{xv\})\cup \{xy\}$. Let us consider how this rewiring affects the clustering coefficient. Specifically, we need only consider the total number of triangles and the total number of length 2 paths in the new graph $G'$.

First, we consider triangles. We note that the only triangles that will be present in $G$ but not in $G'$ are those that have $xv$ as an edge. On the other hand, the only triangles that will be present in $G'$ but not in $G$ are those that have $xy$ as an edge. Hence, we have $N_t(G')=N_t(G)-|N(x, v)|+|N(x, y)|> N_t(G)$.

Likewise, the only length-two paths that appear in $G$ but not $G'$ are those involving the edge $xv$; we note that there are $(d_v-1)+(d_x-1)$ such paths. Similarly, the only length-two paths that appear in $G'$ but not in $G$ are those involving the edge $xy$, of which there are $(d_x-1)+d_y$. Hence, we have $N_p(G')=N_p(G)-(d_v-1)-(d_x-1)+(d_x-1)+d_y = N_p(G)+d_y-d_v+1\leq N_p(G)$.

Therefore, we have that $C(G')=3N_t(G')/N_p(G')>3N_t(G)/N_p(G)=C(G)$.

\end{proof}

We note here that there are many possible variants on the choice of doorway; the fundamental piece of the algorithm is the swing itself. Indeed, in Section \ref{sec:experiments}, we shall also examine the algorithm in a regime in which doorways are chosen randomly, rather than greedily as described in the above algorithms.

\subsection{Local optima}
We now turn to a consideration of local optima under this algorithm. As the algorithm is strictly monotone with respect to clustering, one would expect that any such optima will have a high clustering coefficient. Indeed, as the theorem below shows, this will be the case for our technique.

\begin{theorem}\label{localopt}
Let $G$ be a graph, such that the edges of $G$ can be partitioned into cliques of size at least 3, say $C_1, C_2, \dots, C_k$. Then $G$ is a local optimum with respect to the above algorithm.
\end{theorem}

\begin{proof}
We need only show that there are no legal rewires to be performed under the algorithm. We recall that in order to rewire $vx$ to $yx$, we must have that $|N(x, v)|<|N(x, y)|$, and hence if $N|(x, y)|=0$, there are no edges that can be rewired to $xy$. Moreover, note that any two cliques can share at most one vertex, as the cliques partition the edges of $G$.

Now, suppose that $x, y\in \bar{E}$. Now, if $d_x=0$ or $d_y=0$, then $|N(x, y)|=0$ and there will be no legal rewires.

If not, both $x$ and $y$ appear as members of at least one clique. Note that they are not in the same clique, as if they were, we would have $x\sim y$. If the cliques including $x$ and the cliques including $y$ share no vertices, then $|N(x, y)|=0$, and there will be no edge to rewire to $xy$.

Hence, we may assume that $x\in C_i$, $y\in C_j$, where $|V(C_i)\cap V(C_j)|=1$. Thus, $x$ and $y$ share as a neighbor the vertex at which the two cliques intersect, and hence $|N(x, y)|=1$. Now, note that if $v$ is a neighbor of $x$, then $v$ and $x$ appear in some clique of size at least 3 together, and hence $|N(x, v)|\geq 1$. But then $xv$ cannot be rewired to $xy$. As the same will be true of neighbors of $y$, there is no edge that satisfies the requirements of the algorithm.

Therefore, $G$ is locally optimal with respect to the algorithm.
\end{proof}

We note that the clustering coefficient of these graphs will be quite close to 1. Taking $G$ to be as described in Theorem \ref{localopt}, define $H$ to be the graph on $[k]$, wherein $i\sim j$ if and only if $C_i$ and $C_j$ share a vertex. We then have
\begin{eqnarray*}
C(G) & = & \frac{3N_t(G)}{N_p(G)}\\
& = & \frac{ 3\displaystyle\sum_{i=1}^{k}{n_i\choose 3}}{\displaystyle\sum_{i=1}^{k}n_i{n_i-1\choose 2} + \displaystyle\sum_{i\sim_Hj}n_in_j}\\
& \geq & \frac{\displaystyle\sum_{i=1}^{k}n_i{n_i-1\choose 2} }{\displaystyle\sum_{i=1}^{k}n_i{n_i-1\choose 2} +\sum_{i, j}n_in_j}
\end{eqnarray*}

Clearly, the fewer common vertices we have among cliques, the higher the clustering coefficient will be. Moreover, if we imagine that $k$ is fixed, but the number of vertices in the graph is tending to $\infty$, then $C(G)$ tends to $1$ asymptotically.

To the best of these authors' knowledge, the optimum connected graph on $n$ vertices with a fixed number of edges $m$ with respect to the global clustering coefficient $C(G)$ is unknown. We note that these algorithms as written do not necessarily require that the edge rewiring preserves connectivity or components in the graph $G$, although it is clear based on the structure that this algorithm cannot combine two components into one. However, it is straightforward to construct examples in which the algorithms will rewire a bridge, thus disconnecting a formerly connected graph (one such example is shown in Figure~\ref{Fig:bridge}
). Moreover, we note that there are local optima that do not take the form described above; for example, a barbell graph in which two complete graphs are connected by exactly one edge cannot be partitioned into cliques of size at least 3, but it is still optimal with respect to this rewiring algorithm.

\begin{figure}[H]
\centering
\includegraphics[width = .7\textwidth]{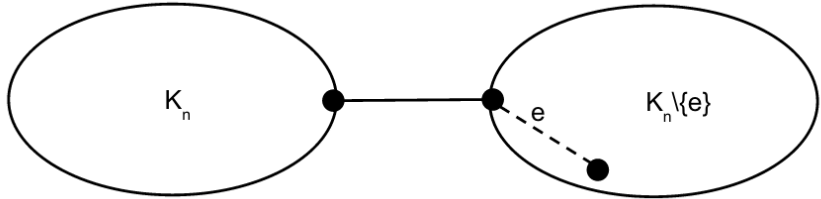}
\caption{A graph $G$ in which the only legal rewiring will disconnect the graph.\\}\label{Fig:bridge}
\end{figure}

Finally, we consider the case of a ring lattice with respect to this algorithm. Define the ring lattice $L(n, k)$ to be the graph with $n$ vertices, labeled as $v_1, v_2, \dots, v_n$, and having $v_i$ adjacent to $v_j$ if and only if $|i-j|\leq k 
\mod n$. We note that in the Watts-Strogatz experiment described in \cite{WattsStrogatz1998}, it is this lattice that is rewired randomly to produce a random graph. As our algorithm increases clustering and decreases path length, essentially reversing the process of the experiment performed by Watts-Strogatz, one might expect that the algorithm cannot improve upon the ring lattice; that is indeed the case.

\begin{theorem}
Let $2\leq k < \frac{n}{4}$. Then $L(n, k)$ is locally optimal with respect to the above algorithms.
\end{theorem}

\begin{proof}
Note that it is sufficient to show that any single edge rewiring on $L(n, k)$ cannot increase the clustering coefficient. To that end, let us suppose that we have an edge $v_iv_j$ that will be rewired to the edge $v_iv_\ell$. Note here that it must be the case that $|i-j|<k$; wolog let us suppose that $i=1$ and $1<j\leq k+1$. Note that the edge $v_1v_j$ in $L(n, k)$ participates in $2k-j$ distinct triangles, namely those triangles with vertices $(v_1, v_j, v_t)$, where $n+(j-k)\leq t\leq n$ and $2\leq t\leq k+1$, provided $t\neq j$. Note moreover that $2k-j\geq k-1$.

Let $L'$ be the graph obtained by rewiring $v_1v_j$ to $v_1v_\ell$. Note that if $n/2<\ell<n-k+1$, then the only triangles that the edge $v_1v_j$ participates in are those for which the third vertex, $v_t$, has $\ell<t\leq n$ and $t-\ell<k, n+1-t<k$. Note that there are at most $k-1$ such triangles, when $\ell=n-k$. Hence, this cannot increase the number of triangles.

Likewise, if $v_1v_j$ is rewired to an edge $v_1v_\ell$ having $\frac{n}{2}\leq\ell>k+1$, we create at most $k-1$ triangles in $L'$.

Hence, we cannot increase the number of triangles that appear in $L'$ by a single edge rewiring. Moreover, if we consider the number of wedges in $L$ and $L'$, we obtain
\begin{eqnarray*}
N_p(L(n, k)) &=& \sum_{v\in V} {\deg(v)\choose 2}\\
& = & n {2k\choose 2}, and
\end{eqnarray*}

\begin{eqnarray*}
N_p(L'(n, k)) & = & \sum_{v\in V}{\deg(v)\choose 2}\\
& = & (n-2){2k\choose 2} + {2k-1\choose 2} + {2k+1\choose 2}
\end{eqnarray*}

It is trivial to compute that $N_p(L')>N_p(L)$. Hence, a rewiring cannot increase the number of triangles, but always increases the number of paths of length 2. Hence, no rewiring can increase the clustering coefficient, and $L(n, k)$ is locally optimal with respect the described algorithms.

\end{proof}

\subsection{Degree sequences}

It is clear from the definition of the algorithms above that the degree sequence in $G$ will not be preserved under these rewirings. Indeed, by examining steps \ref{Deg1} and \ref{Deg2} in the algorithm statements, it can be seen that this algorithm always rewires edges from higher degree nodes to lower degree nodes. Although it might seem, based on this feature, that the rewired graph would tend toward regularity, it can be seen empirically that this is not the case. If we expect that high degree nodes will control large clusters, the edges that are rewired are those that are not actually involved in many triangles with that high degree node. 

This being said, the algorithm can be modified to allow for a degree-preserving version. In this version, rather than rewiring one edge, we must rewire two edges at a time, as follows. 

\begin{algorithm}[H]
\floatname{algorithm}{Degree Sequence Preserving Rewiring Algorithm}
\caption{}
\label{Algthree}
\begin{algorithmic}[1]
\STATE Find a pair $\{x, y\}\in \bar{E}$ with a maximum value of $|N(x, y)|$
\STATE Let $F\subset \bar{E}$ be the set of nonedges in $G$ such that for each $\{u, v\}\in F$, we have $\{u,v\}$ is independent from $\{x, y\}$ and $u\sim x, v\sim y$. Among these, choose a  nonedge $\{u, v\}$ with the maximum value of $|N(u, v)|$.
\STATE \label{oneway} If $|N(u, v)|+|N(x, y)|>|N(x, u)|+|N(v, y)|$ and $|N(u, v)|+|N(x, y)|>|N(u, y)|+|N(x, v)|$, delete the edges $\{x, u\}$ and $\{y, v\}$, and replace them with the edges $\{x, y\}$ and $\{v, y\}$.
\STATE \label{otherway} Otherwise, if $|N(u, y)|+|N(v, x)|>|N(x, u)|+|N(v, y)|$, delete the edges $\{x, u\}$ and $\{y, v\}$, and replace them with the edges $\{x, v\}$ and $\{u, y\}$.
\STATE If neither the conditions of steps \ref{oneway} or \ref{otherway} are met, return to step 2 and remove the edge $\{u, v\}$ from consideration.
\STATE If no edge in $F$ satisfies the requirements, return to step 1 and choose a different edge from $G$.
\end{algorithmic}
\end{algorithm}

We note that this algorithm can be seen as a revision of Algorithm \ref{Algone}; a similar revision can be made for Algorithm \ref{Algtwo}. We further note that the benefit of degree sequence preservation here may be outweighed by the expense of such an algorithm; the computation time is substantially higher, since in step 2, we must consider substantially more edges in $G$ than in the previous versions. As with Algorithms \ref{Algone} and \ref{Algtwo}, it is straightforward to show that this algorithm is monotone with respect to the global clustering coefficient; indeed, the number of length-two paths here is constant, and hence the only change is that we are increasing the total number of triangles in $G$. 

\subsection{Complexity}

Computationally, the most expensive step of both Algorithms \ref{Algone} and \ref{Algtwo} is the first, the calculation of the pair $\{x,y\} \in \bar{E}$ with the maximum (minimum) value of $|N(x,y)|$, which corresponds to finding the maximum (minimum) value of a subset of the elements of $A^2$.  In its most straightforward implementation, $A^2$ can be calculated in $\mathcal{O}(n^3)$ time, where $n$ is the number of vertices in $G$.  However, for many large, real-world networks this computational time can be prohibitive.  Algorithms \ref{Algone} and \ref{Algtwo} don't require the full calculation of $A^2$, only lists of the number of triangles and wedges in which each node participates.  The number of wedges can be calculated in $\mathcal{O}(n)$ time using the degree of each vertex. In \cite{Co09}, an elegant algorithm for efficiently enumerating all triangles in a network using MapReduce was introduced and, in \cite{NoWiPhBe10}, it was shown that this enumeration can be done in $\mathcal{O}(n)$ for power law graphs with an exponent of more than $\frac{7}{3}$ and a maximum vertex degree bounded by $\sqrt{n}$.  Additionally, the authors in \cite{NoWiPhBe10} provide experimental results which indicate that the calculation remains $\mathcal{O}(n)$ even when the max degree is bounded by $n-1$.   

Isolating the appropriate subset and finding its maximum (minimum) can be done in $\mathcal{O}(n)$ time.  Each iteration of the algorithm will change the number of triangles and wedges in the network.  However, it will only change the values of vertices $x, y, u,$ and $v$ and those of nodes in their neighborhoods.  These effects can be calculated {\it a priori} and the changes can be implemented through neighborhood-centric updates which involve summing (or subtracting) a few appropriate values, which takes $\mathcal{O}(n)$ time. Thus, the full enumeration of the number of triangles in the network only needs to be done once.

The rest of the algorithms consist of find the vertex $v \in N(x) \Delta N(y)$ with minimum $f_v$, checking that it meets the degree requirements, moving the edge and updating the lists of the numbers of wedges and triangles in which each vertex participates.  These can all be done in $\mathcal{O}(n)$ time for each iteration.  Thus, for the majority of real-world networks, the total cost of the algorithm is $\mathcal{O}(kn)$, where $n$ is the number of vertices in the graph and $k$ is the number of edges to be rewired. 

\section{Results on Generated Networks}
\label{sec:experiments}

We illustrate the effects of the rewiring algorithms by running them on simulated networks. We create synthetic networks of various types and then iteratively rewire them using one of the rewiring procedures until there are no valid rewiring moves left.  These experiments are run on Erd\"os-R\'enyi ($G_{n,p}$) networks \cite{ErRe59} and Barab\'asi-Albert preferential attachment networks \cite{BaAl99}. 

Both rewiring algorithms 
led to the global clustering coefficient of the network increasing monotonically. 
The ``Swing Toward Best" algorithm (Algorithm \ref{Algone}) had generally greater clustering gains than the ``Swing Away from Worst" algorithm (Algorithm \ref{Algtwo}), as can be seen in Figure~\ref{Fig:algorithm_comparison} on Erd\"os-R\'enyi $G_{n,p}$ networks with $n=100$ and $p=0.07$. This is expected, as the ``Swing Toward Best" 
algorithm is a global optimization procedure while the 
``Swing Away from Worst" algorithm is a local optimization procedure. Results from choosing edges to rewire randomly or probabilistically based on the two procedures are also plotted.  

It is easy to see how employing different methods for selecting the candidate edges lead to clustering gains of different sizes even within the ``Swing Toward Best'' scheme. The greedy versions of the algorithms raised the clustering the most (see the blue line on the left of Figure~\ref{Fig:algorithm_comparison}). Randomly selecting a candidate edge still increased clustering, but not as rapidly (see the red line on the left of Figure~\ref{Fig:algorithm_comparison}). In between was selecting candidate edges probabilistically (see the green line on the left of Figure~\ref{Fig:algorithm_comparison}).  Similar results can be seen within the ``Swing Away from Worst'' scheme, although random and probabilistic edge selection perform much more similarly to the optimum in this case (see the right side of Figure~\ref{Fig:algorithm_comparison}).

\begin{figure}[h]
\centering
\includegraphics[width=0.8\textwidth]{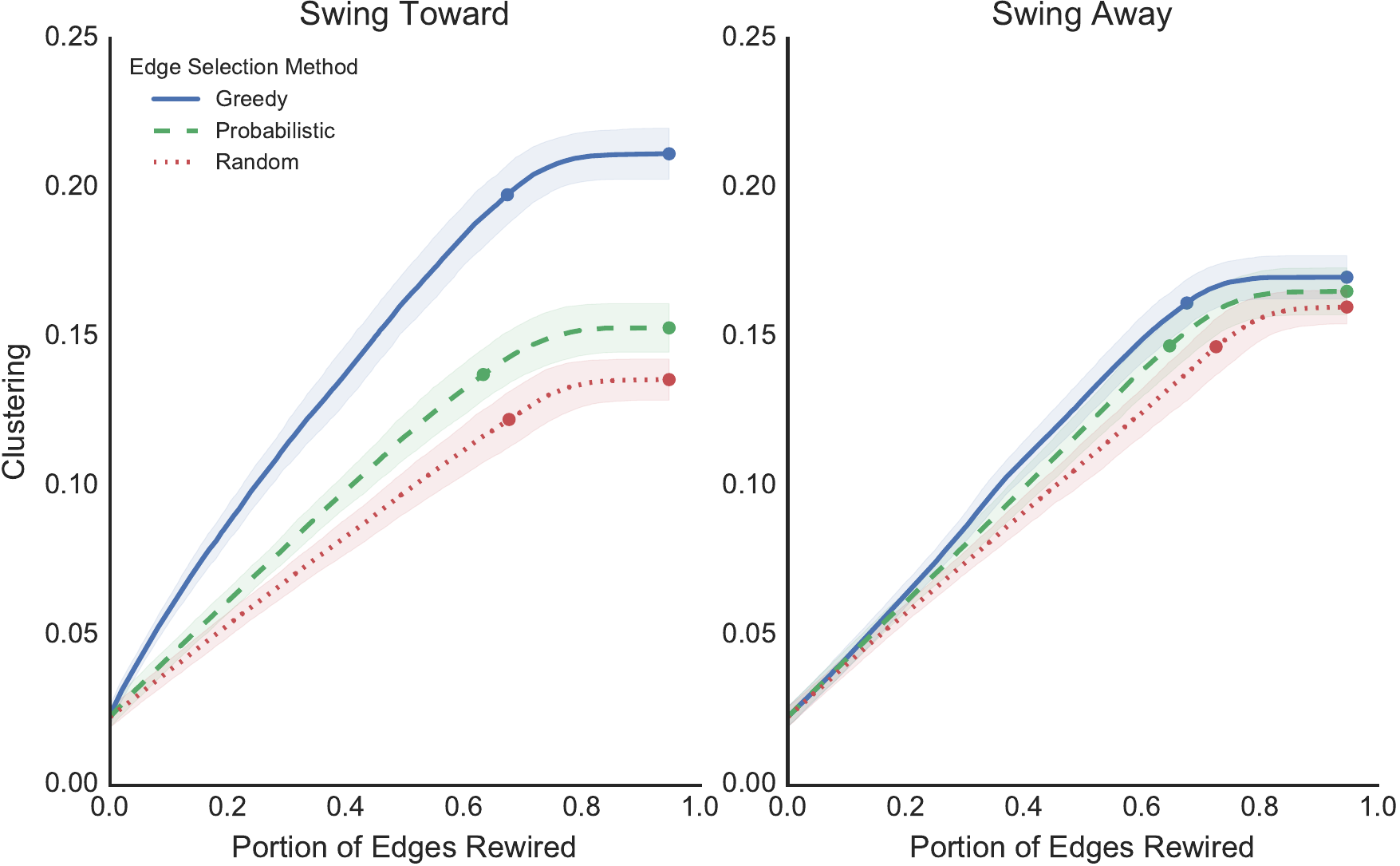}
\caption{\textbf{Rewiring a network's edges using either algorithm \ref{Algone} or \ref{Algtwo} raised its clustering coefficient.}
Simulated networks: 100 replications of Erd\"os-R\'enyi ($G_{n,p}$) networks ($n=100, p=0.07$). (Unless otherwise stated, for all figures, lines show the mean over ten runs of the algorithm; shading shows the standard deviation; the first dot is the point at which the first network stopped rewiring; the second dot is point at which the last network stopped rewiring.)
Different line colors/line styles indicate different methods for finding an edge to rewire.
This figure shows the outcomes of the ``Swing Toward" rewiring algorithm \ref{Algone} (left panel) and the ``Swing Away" rewiring algorithm \ref{Algtwo} (right panel).
}\label{Fig:algorithm_comparison}
\end{figure}

For the rest of the paper, all results are generated using the optimum version of the ``Swing Toward" algorithm, Algorithm~\ref{Algone}, unless otherwise specified.

In Figure~\ref{Fig:algorithm_comparison}, it can be seen that in all 100 runs, over 60\% of the edges in the Erd\"os-R\'enyi networks were rewired before the algorithm stopped and, in some cases nearly 100\% were.  It is unclear {\em a priori} whether algorithmic restrictions to preserve degree distribution will also limit the number of valid rewirings.  In Figure~\ref{Fig:degree_preservation}, the degree-preserving version of the optimum ``Swing toward Best'' algorithm is run on 100 instances of Erd\"os-R\'enyi $G_{n,p}$ networks with $n=100$ and $p=0.07$.  It can be seen that, in every instance, over 60\% of the edges were rewired before the algorithm completed, even with the modification to preserved degrees. While this procedures put greater requirements on what constituted a valid rewiring move, it did not seem to limit the number of edges that were rewired. However, it did not produce as large of an increase in clustering, which was expected due to the more limited options of where edges could be moved.

\begin{figure}[h]
\centering
\includegraphics[]{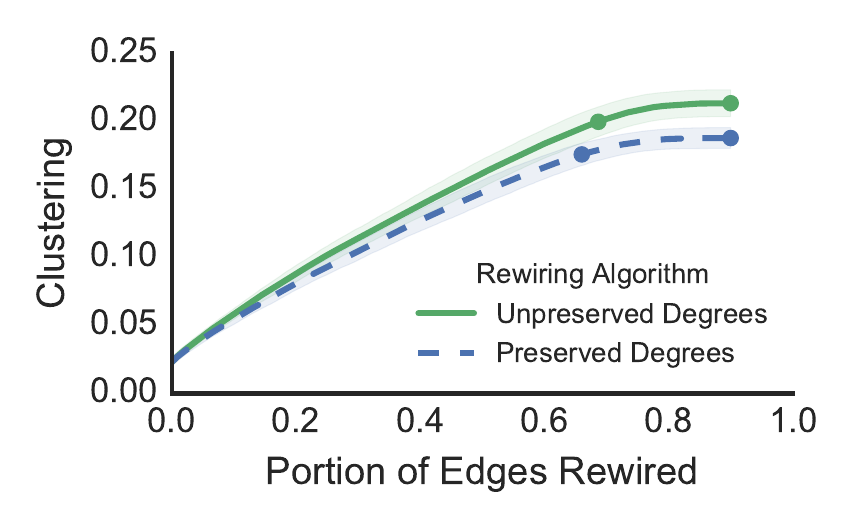}
\caption{\textbf{Rewiring a network while preserving the degree sequence also increased clustering, though not as quickly.} 
Simulated networks: 100 replications of Erd\'os-R\'enyi ($G_{n,p}$) networks ($n=100, p=.07$).
The green solid line shows rewiring with the ``Swing Toward" Algorithm \ref{Algone}, which does not preserve the degree sequence.
The blue dashed line shows rewiring with Algorithm \ref{Algthree}, which does preserve the degree sequence. 
Although both algorithms increased the clustering coefficient of the network, Rewiring Algorithm \ref{Algone}, which allowed for changes in the degree sequence of the network, increased the clustering more than Algorithm \ref{Algthree}. \\}\label{Fig:degree_preservation}
\end{figure}

\begin{figure}[h]
\centering
\includegraphics[]{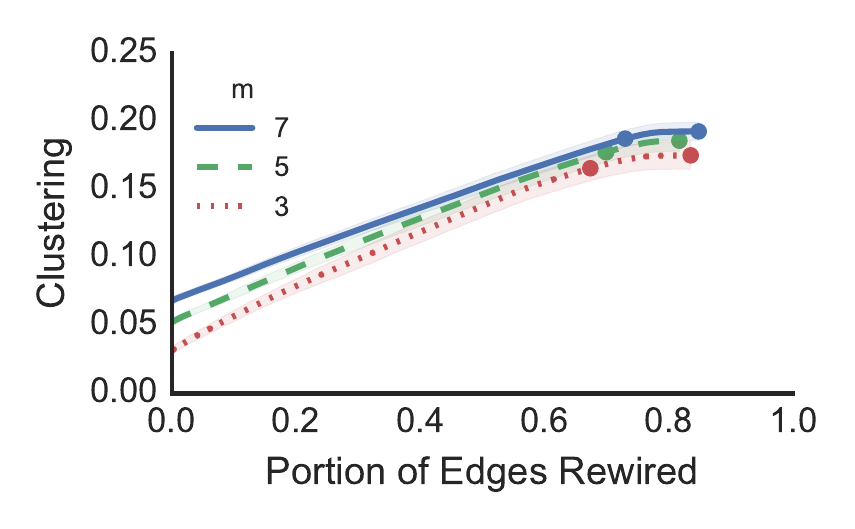}
\caption{\textbf{Rewiring raised the clustering of scale-free networks.}
Simulated networks: 100 replications of 100 node Barab\'asi-Albert networks.
Different line colors/styles correspond to different values of $m$, the number of nodes that a new node attaches to during the initial network generation.
}\label{Fig:Barabasi_Albert}
\end{figure}

We also ran the optimal version of the ``Swing Toward'' algorithm on Barab\`asi-Albert preferential attachment networks with a variety of parameters.  These results can be seen in Figure~\ref{Fig:Barabasi_Albert}.  Algorithm \ref{Algone} was run on 100 instances each of preferential attachment networks with $m=3,5,$ and $7$, where $m$ is the number of edges each new node added begins with.  Although the initial clustering coefficient was different for each of these three parameters, with a lower $m$ corresponding to a lower initial clustering, after the rewiring the clusterings converge to values that are much closer.  That is, the networks with a lower $m$ show greater gains in clustering coefficient than those with a higher $m$.  Again, we see that, in every instance, over 60\% of the edges are rewired before the algorithm halts.  We also see that the final clustering coefficient of the rewired networks hovers around 0.15-0.2, which is the same range as the final clustering coefficient in the rewired Erd\"os-R\'enyi networks.

\subsection{Different Kinds of Clustering}
In Section \ref{sec:algorithms}, we prove that these rewiring algorithms are guaranteed to increase a network's global clustering coefficient, but there is another common measure of a network's clustering: the average local clustering coefficient. The rewiring algorithms also typically increase average local clustering, as can be seen in Figure \ref{Fig:average_local_clustering} (100 runs of rewiring Erd\"os-R\'enyi $G_{n,p}$ networks with $n=100$ and $p=.07$). However, it is possible for the local clustering to stall or go down during portions of the rewiring process, whereas for total clustering it is not possible. 

\begin{figure}[h]
\centering
\includegraphics[width = 0.5\textwidth]{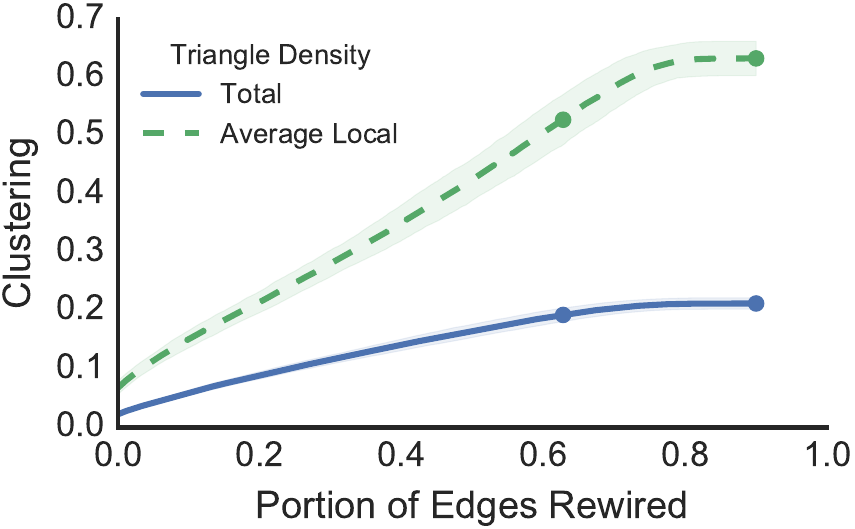}
\caption{\textbf{Rewiring networks increased different kinds of clustering.} The change in average local clustering over 100 runs on an Erd\"os-R\'enyi $G_{n,p}$ network ($n=100, p=.07$).
The blue solid line shows the total triangle density (global clustering coefficient).
The green dashed line shows the average local triangle density (average local clustering coefficient).
}\label{Fig:average_local_clustering}
\end{figure}

\subsection{Small World Creation}
In addition to clustering, the rewiring algorithms introduced here also alter other properties of the network structure. One such property is the average path length of the network that is undergoing rewiring. We examined this in more detail and found that the average path length increased during rewiring; as clusters formed it became harder to quickly navigate between them.  This can be seen in the left panel of Figure~\ref{Fig:path_length_small_world}
, which shows these changes on 100 runs of the optimal ``Swing Toward'' rewiring algorithm on Erd\"os-R\'enyi $G_{n,p}$ networks with $n=100$ and $p=0.07$.  However, the trade-off between clustering and path length was not constant. Clustering increased faster than path length during the majority of the edge rewires but, at the end of the rewiring process, the path length increased more quickly and the clustering coefficient stabilized. A network with a high clustering and low path length is commonly known as a small-world network \cite{WattsStrogatz1998} and the {\em small-world index} summarizes this relationship through the ratio of the clustering coefficient and the path length. A high small-world index indicates that a network has a particularly complex structure.  We plot the effect of optimal ``Swing Toward'' rewirings on the left of Figure~\ref{Fig:path_length_small_world}. The small-world index increased during the rewiring, as the clustering grew faster than the path length, but then plateaued and slightly decreased before the rewiring algorithm terminated.

\begin{figure}[h]
\centering
\includegraphics[width = \textwidth]{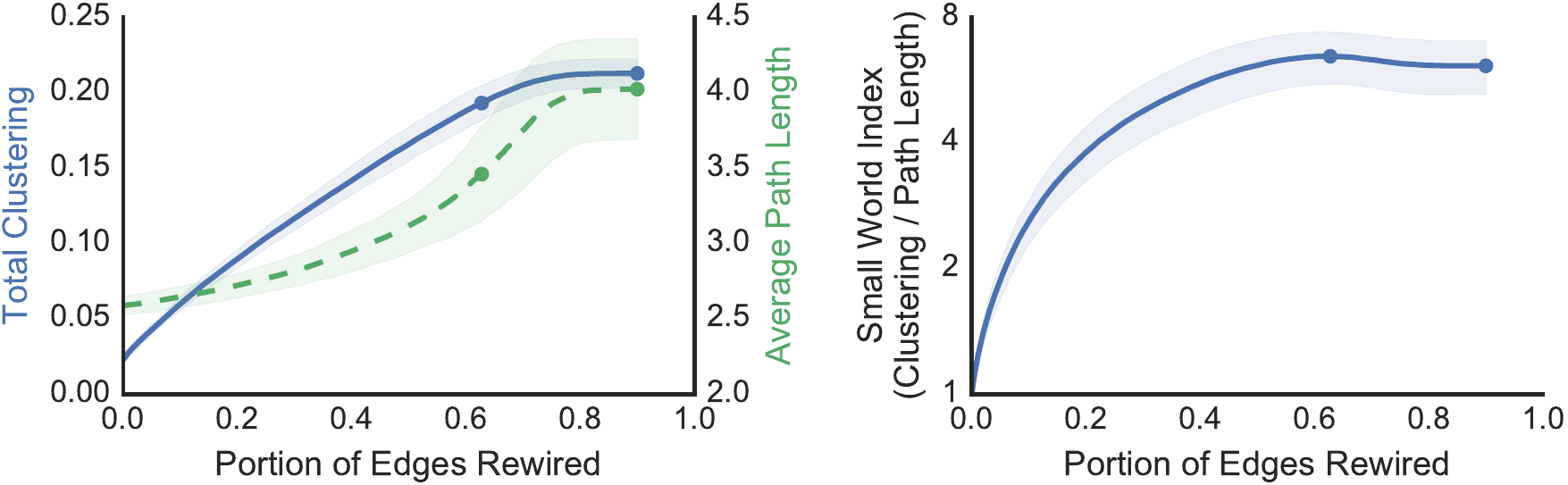}
\caption{\textbf{The rewiring algorithm increases clustering faster than it increases path length, creating a small-world network.} We show the effects of optimal ``Swing Toward'' rewiring on 100 replications of Erd\"os-R\'enyi $G_{n,p}$ networks with $n=100$ and $p=.07$.
On the left, the change in both the global clustering coefficient (blue solid line) and the average path length (green dotted line) are plotted.
On the right, the change in the small world index, the ratio of the clustering coefficient and the average path length, is plotted. 
Initially, average path length increases along with the clustering coefficient. However, because clustering increases faster than path length, the rewiring process increases the small-world index by several multiples relative to the initial value.
}\label{Fig:path_length_small_world}
\end{figure}

This rewiring procedure can be thought of as the reverse of the classic algorithm for creating a small world network, introduced by Watts and Strogatz in \cite{WattsStrogatz1998}. In that process, the network starts as a lattice and is rewired randomly. Randomizing every edge turns the network into an Erd\"os-R\'enyi network, but randomizing only a small percentage of the edges greatly decreases the network's mean path length while only slightly decreasing the clustering, yielding a small world structure. Figure \ref{Fig:Watts_Strogatz_Inverse:a} shows this behavior on a ring-lattice with 100 nodes and degree of 6. As more edges are randomized, the network's average local clustering and path length both decrease, although at different rates, eventually terminating in a random graph.

\begin{figure}[!h]
    \centering
    \begin{subfigure}[h]{\textwidth}
            \centering
            \includegraphics[width=0.6\textwidth]{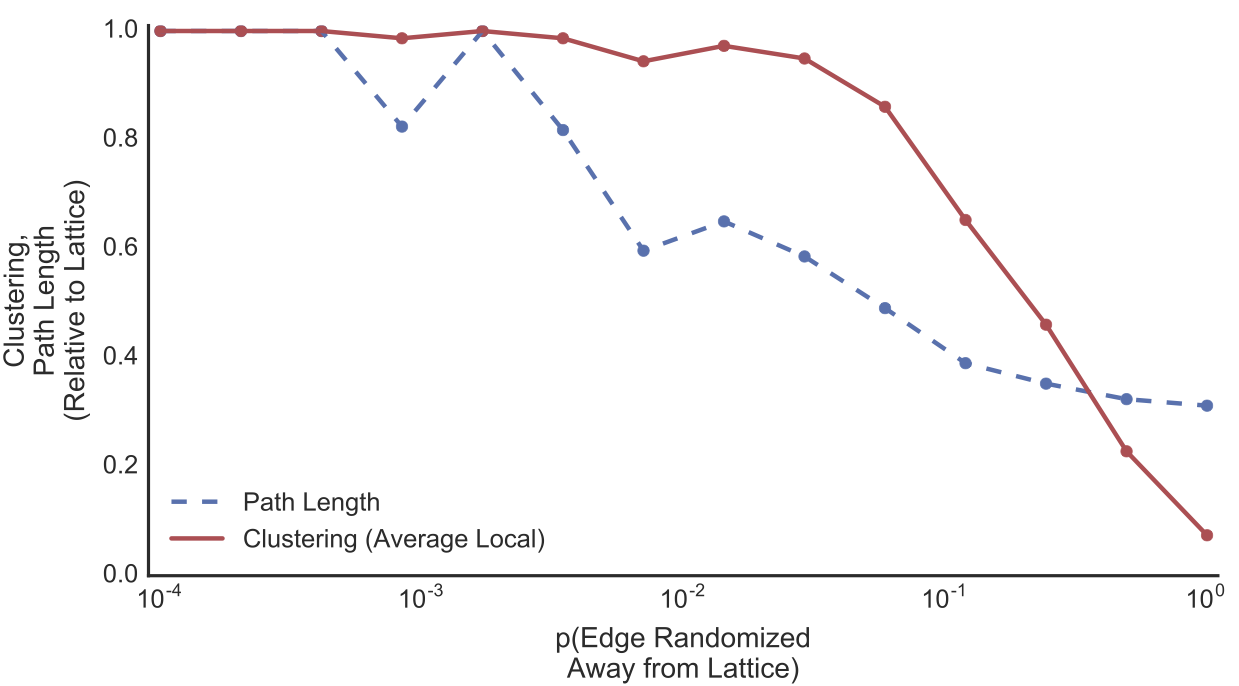}
            \caption{Changes in average local clustering (red solid) and path length (blue dashed) relative to the original ring lattice as a greater percentage of the edges of a ring lattice ($n=100$, $k=6$) are randomly rewired.}
    \label{Fig:Watts_Strogatz_Inverse:a}
    \end{subfigure}
    
\begin{subfigure}[h]{\textwidth}
            \centering
            \includegraphics[width=0.6\textwidth]{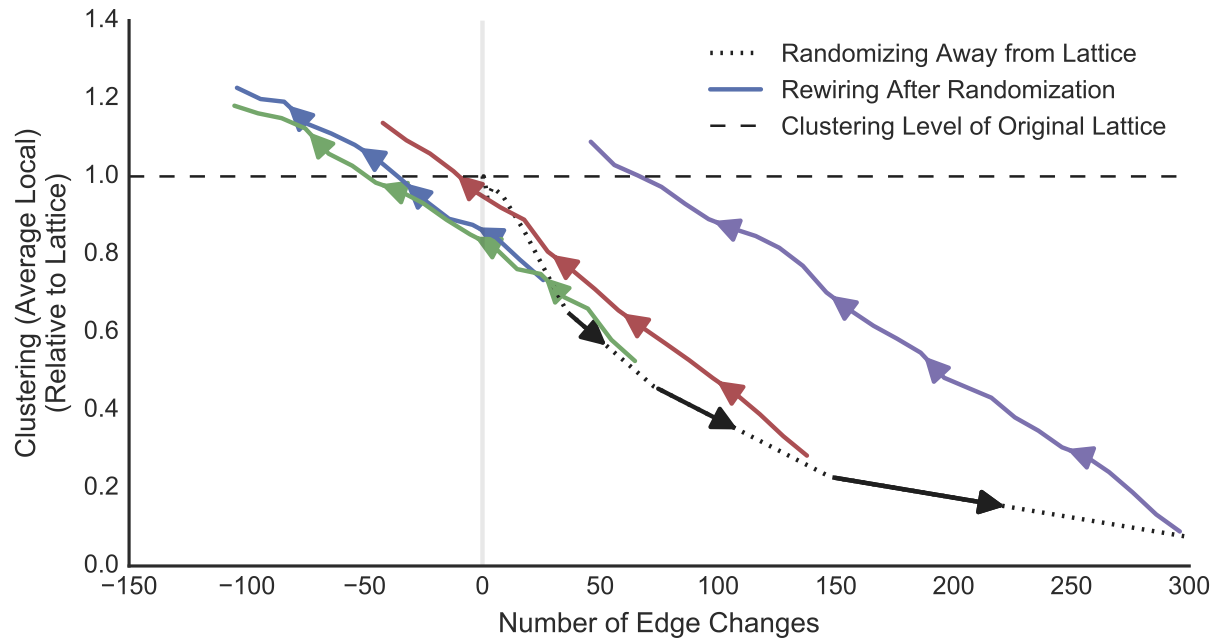}
            \caption{Changes in average local clustering relative the original ring lattice as edges are rewired randomly (black dotted) or rewired using Algorithm 1 after various levels of edge randomization: 300 edges (purple solid), 150 (red solid), 75 (green solid), and 50 (blue solid).} 
    \label{Fig:Watts_Strogatz_Inverse:b}
    \end{subfigure}
    
\begin{subfigure}[h]{\textwidth}
            \centering
            \includegraphics[width=0.6\textwidth]{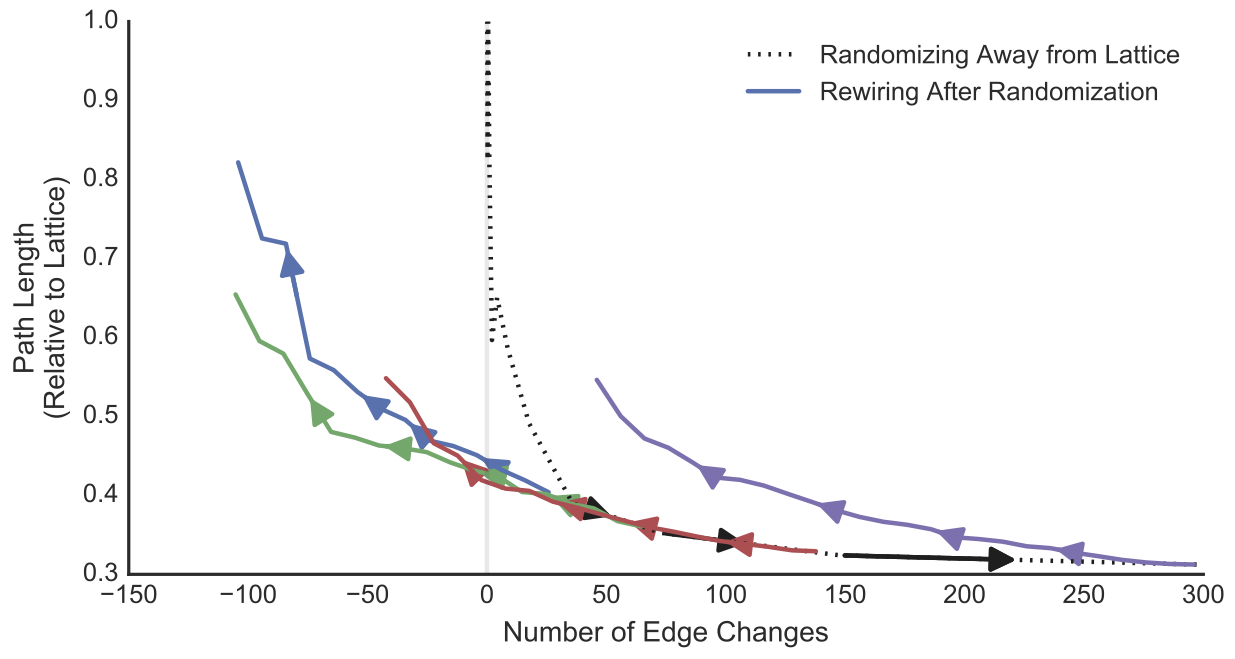}
            \caption{Changes in average path length relative the original ring lattice as edges are rewired randomly (black dotted) or rewired using Algorithm 1 after various levels of edge randomization: 300 edges (purple solid), 150 (red solid), 75 (green solid), and 50 (blue solid). }
    \label{Fig:Watts_Strogatz_Inverse:c}
    \end{subfigure}
    \caption{\textbf{Rewiring reversed the clustering loss of the Watts-Strogatz algorithm.}
    }
\label{Fig:Watts_Strogatz_Inverse}
\end{figure}

The rewiring algorithms introduced here can be thought of as inverting this randomization, creating clustering instead of destroying it.  After randomizing different portions of the original lattice, the rewiring algorithms presented above are employed to recreate clustering in the network, moving it away from a random graph, as seen in Figure \ref{Fig:Watts_Strogatz_Inverse:b}.  Here, each colored line with arrows pointing left represents starting the rewiring at a different level of randomization. The clustering can end up higher than that of the original lattice, both due to the randomization and because the rewiring altered the network's degree distribution.  Because of this the rewiring algorithm is not a true inverse of the randomization, but is indeed a reversal of the clustering loss.

As discussed earlier, the rewiring process increases the network's average path length, but when it terminates, the average path length is still lower than that of the original lattice. This can be seen in Figure \ref{Fig:Watts_Strogatz_Inverse:c}. Thus, we see that the rewiring algorithm creates a small world network, with high clustering and low path length relative to the lattice.

\section{Results on Real-World Networks}
\label{sec:real_experiments}

We further applied our algorithm to several real-world networks. These networks were formed from data collected by Traud et al.~and analyzed in \cite{facebookold} and \cite{facebooklong}. They present a snapshot of the Facebook network on a single day in September of 2005. At that time, Facebook was a fairly new entity, initially established as a Harvard-exclusive site called ``The Facebook" in February of 2004 and growing to include many colleges by September of 2005 under its current title of Facebook. At the time the data was collected, Facebook members needed to have a \verb|.edu| email address and thus only students and other members of the college's community, including faculty, were able to join \cite{Boyd_FB,Mayer2008329,Lewis2008330}. We have considered the Facebook networks of the California Institute of Technology (Caltech) and Reed College, two of the 100 American colleges and universities on Facebook at the time. While Facebook allowed for cross-college connections between individuals in 2005, the networks to which we apply our algorithm are the largest connected components of the Caltech and Reed networks excluding cross-college connections. 

The Caltech Facebook network, which was the smallest Facebook network in the snapshot, has 762 members (nodes) with 16651 connections (edges) between them and a global clustering coefficient of approximately 0.0971. 
As in Section \ref{sec:experiments}, the edges in this network were rewired using Algorithm~\ref{Algone}. Figure~\ref{Fig:clustering_caltech} shows the monotonic increase in the average local clustering coefficient of the Caltech network from the application of this algorithm.  A total of 10466 rewires were completed before the algorithm could no-longer find an edge to move, resulting in a network that was approximately 62.8\% rewired with a clustering coefficient of 0.230.

\begin{figure}[h!]
    \centering
    \includegraphics[width = 0.5\textwidth]{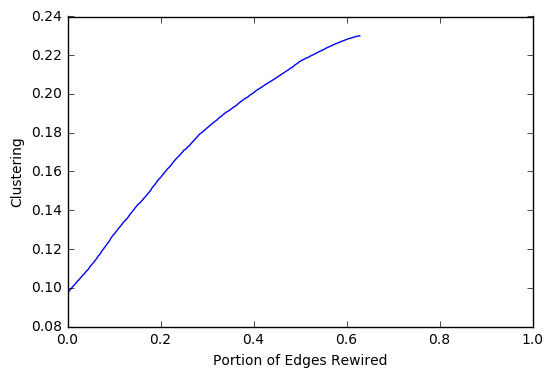}
    \caption{\textbf{Rewiring the Caltech Facebook network increases clustering.}}
    \label{Fig:clustering_caltech}
\end{figure}

Along with altering the clustering, the rewirings performed in Algorithm~\ref{Algone} affect other network properties, such as the average shortest path-length.  The change in average shortest path-length as edges are rewired in the Caltech Network compared to the change in clustering is shown in Figure~\ref{Fig:path_caltech_every100}. Measuring the average shortest path-length across the Caltech network throughout the rewiring process is computationally slow, so network properties were calculated every 100 rewiring steps with a total of 10400 rewiring steps taken before the graph became disconnected.  As was seen in the generated examples in Section \ref{sec:experiments}, although the average shortest path-length increases, in the initial rewiring steps it does so much more slowly than the average local clustering coefficient.  Similarly to earlier, when Algorithm~\ref{Algone} is used to rewire a small percentage of the edges in the network, the average local clustering coefficient increases significantly while other network properties are minimally changed. 

\begin{figure}[h!]
    \centering
    \includegraphics[width = 0.5\textwidth]{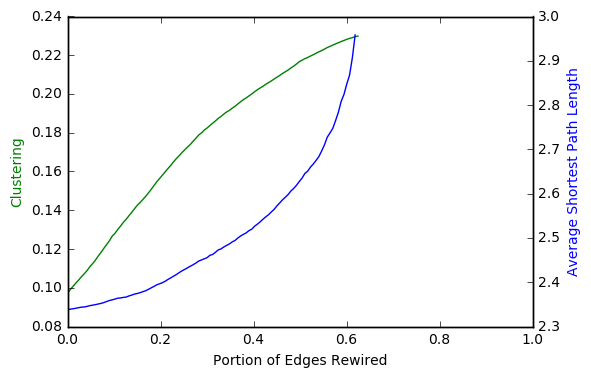}
    \caption{\textbf{Clustering (green) and path length (blue) of Caltech Facebook network with network properties being calculated every 100 rewiring steps.} }
    \label{Fig:path_caltech_every100}
 \end{figure}

A second social network examined was the Reed College Facebook network.  This network has 962 members with 18812 connections between them and a global clustering coefficient of approximately 0.0736, slightly lower than that of the Caltech network. 
Figure~\ref{Fig:clustering_reed} shows how the clustering of the Reed College network changed when Algorithm~\ref{Algone} was applied and an increasing number of edges were rewired. A total of 14762 rewires were completed before the algorithm could no-longer find an edge to move, resulting in a network that was approximately 78.5\% rewired with a clustering coefficient of 0.197. 

\begin{figure}[h!]
    \centering
    \includegraphics[width = 0.5\textwidth]{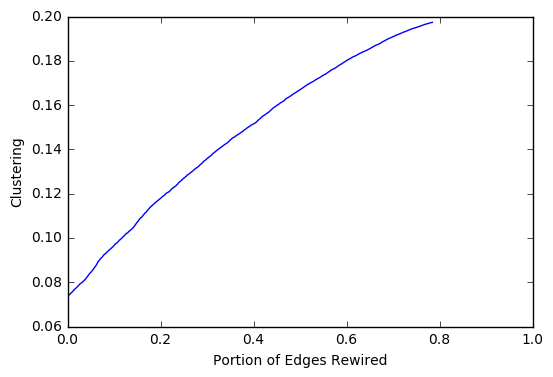}
    \caption{\textbf{Rewiring the Reed College Facebook network increases clustering.}}
    \label{Fig:clustering_reed}
\end{figure}

Similar treatment given to the Caltech network can be applied to the Reed network, computing average shortest path-length at every 500 rewires due to time considerations, until a rewire can no longer be completed. These results are shown in Figure~\ref{Fig:path_reed_every500}.  As seen in Figure~\ref{Fig:path_caltech_every100}, in the early stages of rewiring the average local clustering increases much more quickly than the average shortest path-length, again increasing the ``small-worldness'' of the network.

\begin{figure}[h!]
    \centering
    \includegraphics[width = 0.5\textwidth]{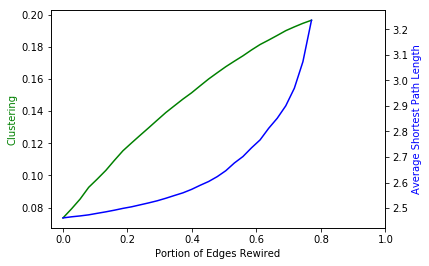}
    \caption{\textbf{Clustering (green) and path length (blue) of Reed Facebook network with network properties being calculated every 500 rewiring steps.} } 
    \label{Fig:path_reed_every500}
\end{figure}

It may seem surprising that both of these networks were able to rewire more than half (and in the case of the Reed network, almost 80\%) of their edges. Yet, it is important to note that these graphs likely reflect incomplete social networks due to the novelty of Facebook at the time \cite{facebooklong,Boyd_FB}. This would leave many missing connections between close friends, allowing the algorithm to exploit the missing links, rearrange edges, and increase clustering. 

\section{Concluding Remarks}
\label{sec:conclusions}

We have introduced a new set of algorithms for rewiring edges in a network with the goal of maximally increasing the global clustering coefficient and minimally impacting other network properties, with a focus on preserving degree distribution and average degree length.  Several variations of this algorithm were implemented, including one that preserves the degree sequence of the original network.  We proved that the algorithms strictly increase the global clustering coefficient of a network, provided examples of local optima under Algorithm \ref{Algone}, and discussed the time complexity of running the algorithms.

We ran the algorithms on a number of Erd\'os-R\'enyi ($G_{n,p}$) networks to examine how the clustering coefficient increased both as a function of the the number of nodes in the network (holding the edge density constant) and as a function of the number of edges rewired.  We also examined how these were affected by the exact variation of the algorithm used.  Other experiments examined how the increase in the average path length or the average local clustering coefficient compared to that of the global clustering coefficient.  Our experiments corroborate the theoretical findings that the algorithms presented strictly increase the global clustering coefficient. They also show the average local clustering coefficient increases, though not necessarily monotonically.

Additionally, the simulation experiments demonstrate that the rewiring procedures create a small world. This is essentially the reverse of the process that Watts and Strogatz used when they introduced the small world concept \cite{WattsStrogatz1998}. They started with a regular lattice and rewired edges randomly until the network was fully random. Along the way the path length decreased faster than the clustering decreased, and so there was a period during the rewiring in which clustering was high and path length was low: a small world. Here, our rewiring procedures are doing the reverse by starting with a randomized network and deliberately rewiring edges to create a more regular network. And again, along the way a small world is created. Thus, the rewiring algorithms introduced here give a reciprocal story to how to build a small world network, given a fixed number of nodes and edges. These rewiring algorithms may be useful to those seeking to grow small world networks in engineered physical systems, or to explain how naturally-occurring physical systems construct themselves into small worlds.

We also ran Algorithm \ref{Algone} on several networks from a single-time snapshot of Facebook data \cite{facebookold,facebooklong}. The results from these experiments indicate that the real-world networks behave similarly to the Erd\'os-R\'enyi networks in that, while both clustering and average shortest path-length increase as the number of rewirings increase, when only a small number of edges are rewired the average local clustering increases significantly more quickly than the average shortest path-length. 

Future work includes determining whether using approximations of the initial values of $A^2$ and $N(x,y)$ (the most computationally expensive element of the algorithm) affect the performance of the algorithm.  One aspect of these algorithms which, for the sake of brevity, was not discussed in this paper is that of applications where these algorithms may be deployed.  These include areas such as community detection and graph partitioning.  Future work concerns research into these applications.  Another important aspect of future work is determining, both experimentally and theoretically, the ideal number of edges which should be rewired for various applications in which modifying a network to increase its clustering coefficient is desirable.

\section*{Acknowledgments}
J.A. was supported by the SUTD-MIT Postdoctoral Programme. The work of C.K. was performed under the auspices of the U.S. Department of Energy by Lawrence Livermore National Laboratory under Contract DE-AC52-07NA27344.


\bibliographystyle{siam}  

\end{document}